\documentclass{sig-alternate}

\usepackage{amsthm}
\usepackage{amssymb,amsmath}
\usepackage{graphics}
\usepackage{amsfonts}
\usepackage{mathrsfs}
\usepackage{amscd}
\usepackage{comment}
\usepackage{color}
\usepackage{url}
\usepackage{bbm}
\usepackage[plainpages=false,pdfpagelabels,colorlinks=true,citecolor=blue,hypertexnames=false]{hyperref}

\newtheorem{theorem}{Theorem}

\newtheorem{corollary}[theorem]{Corollary}
\newtheorem{lemma}[theorem]{Lemma}

\newtheorem{defi}[theorem]{Definition}
\newtheorem{example}[theorem]{Example}

\def\qed{\quad\rule{1ex}{1ex}}

\let\set\mathbbm
\def\K{\set K}
\def\O{\mathrm{O}}

\def\deg{\operatorname{deg}}

\def\rf#1#2{#1^{\overline{#2}}}

\begin{document}

\title{Order-Degree Curves for Hypergeometric Creative~Telescoping}

\numberofauthors{2}

\author{%
 \alignauthor Shaoshi Chen\titlenote{Supported by the National Science Foundation (NFS) grant CCF-1017217.}\\[\smallskipamount]
      \affaddr{\strut Department of Mathematics}\\
      \affaddr{\strut NCSU}\\
      \affaddr{\strut Raleigh, NC 27695, USA}\\[\smallskipamount]
      \email{\strut schen21@ncsu.edu}
 \alignauthor \strut Manuel Kauers\titlenote{Supported by the Austrian Science Fund (FWF) grant Y464-N18.}\\[\smallskipamount]
      \affaddr{\strut RISC}\\
      \affaddr{\strut Johannes Kepler University}\\
      \affaddr{\strut 4040 Linz, Austria}\\[\smallskipamount]
      \email{\strut mkauers@risc.jku.at}
}

\maketitle
\begin{abstract}
    Creative telescoping applied to a bivariate proper hypergeometric term
    produces linear recurrence operators with polynomial coefficients, called
    telescopers. We provide bounds for the degrees of the polynomials appearing
    in these operators. Our bounds are expressed as curves in the {$(r,d)$-plane}
    which assign to every order~$r$ a bound on the degree~$d$ of the telescopers. 
    These curves are hyperbolas, which reflect the phenomenon that higher order 
    telescopers tend to have lower degree, and vice versa.
\end{abstract}

\kern-\medskipamount

\category{I.1.2}{Computing Methodologies}{Symbolic and Algebraic Manipulation}[Algorithms]

\kern-\medskipamount

\terms{Algorithms}

\kern-\medskipamount

\keywords{Symbolic Summation, Creative Telescoping, Degree Bounds}

\kern-\medskipamount

\section{Introduction}

We consider the problem of finding linear recurrence equations with polynomial
coefficients satisfied by a given definite single sum over a proper
hypergeometric term in two variables. This is one of the classical problems
in symbolic summation. Zeilberger~\cite{zeilberger90} showed that such a
recurrence always exists, and proposed the algorithm now named after him for
computing one~\cite{zeilberger90a,zeilberger91}. Also explicit bounds are known
for the order of the recurrence satisfied by a given
sum~\cite{wilf92a,mohammed05,bostan10}. Little is known however about the
degrees of the polynomials appearing in the recurrence. These are investigated
in the present paper.

Ideally, we would like to be able to determine for a given sum all the pairs
$(r,d)$ such that the sum satisfies a linear recurrence of order~$r$ with
polynomial coefficients of degree at most~$d$. This is a hard question which we
do not expect to have a simple answer. The results given below can be viewed as
answers to simplified variants of the problem. One simplification is that we
restrict the attention to the recurrences found by creative
telescoping~\cite{zeilberger91}, called ``\emph{telescopers}'' in the symbolic summation
community (see Section~\ref{sec:hg} below for a definition). The second
simplification is that instead of trying to characterize all the pairs $(r,d)$,
we confine ourselves with sufficient conditions.

Our main results are thus formulas which provide bounds on the degree~$d$
of the polynomial coefficients in a telescoper, depending on its order~$r$. The
formulas describe curves in the $(r,d)$-plane with the property that for every
integer point $(r,d)$ above the curve, there is a telescoper of order~$r$ with
polynomial coefficients of degree at most~$d$. As the curves are hyperbolas,
they reflect the phenomenon that higher order recurrence equations may have
lower degree coefficients. This feature can be used to derive a complexity
estimate according to which, at least in theory, computing the minimum order
recurrence is more expensive than computing a recurrence with slightly higher
order (but drastically smaller polynomial coefficients). This phenomenon is
analogous to the situation in the differential case, which was first analyzed by
Bostan et al.~\cite{bostan07} for algebraic functions, and recently for
integrals of hyperexponential terms by the authors~\cite{chen11}.

Our analysis for non-rational proper hypergeometric input
(Section~\ref{sec:trans}) follows closely our analysis for the differential
case~\cite{chen11}. It turns out that the summation case considered here is
slightly easier than the differential case in that it requires fewer cases to
distinguish and in that the resulting degree estimation formula is much simpler
than its differential analogue. For rational input (Section~\ref{sec:rat}), we
derive a degree estimation formula following Le's algorithm for computing
telescopers of rational functions~\cite{abramov02,le03}.

\section{Proper Hypergeometric Terms and Creative Telescoping}\label{sec:hg}

Let $\K$ be a field of characteristic zero and let $\K(n,k)$ be the field of rational functions
in $n$ and~$k$. We will be considering extension fields $E$ of $\K(n,k)$ on which two
isomorphisms~$S_n$ and $S_k$ are defined which commute with each other,
leave every element of $\K$ fixed, and act on $n$ and $k$ via $S_n(n)=n+1$,
$S_k(n)=n$, $S_n(k)=k$, $S_k(k)=k+1$.
A \emph{hypergeometric term} is an element~$h$ of such an extension field~$E$ with
$S_n(h)/h\in\K(n,k)$ and $S_k(h)/h\in\K(n,k)$. Proper hypergeometric terms are hypergeometric
terms which can be written in the form
\begin{equation}\label{eq:hgdef}
 h=p\, x^n y^k \! \prod_{m=1}^M\frac{\Gamma(a_m n+a'_m k+a''_m)\Gamma(b_mn-b'_m k+b''_m)}
  {\Gamma(u_mn+u'_mk+u''_m)\Gamma(v_mn-v'_mk+v''_m)},
\end{equation}
where $p\in\K[n,k]$, $x,y\in\K$, $M\in\set N$ is fixed,
$a_m$, $a'_m$, $b_m$, $b'_m$, $u_m$, $u'_m$, $v_m$, $v'_m$ are nonnegative integers,
$a''_m$, $b''_m$, $u''_m$, $v''_m\in\K$ and the expressions $x^n$,
$y^k$, and $\Gamma(\ldots)$ refer to elements of $E$ on which $S_n$ and $S_k$
act as suggested by the notation, e.g.,
\begin{alignat*}1
  S_n(\Gamma(2n-k+1))&=(2 n{-}k{+}1)(2 n{-}k{+}2)\Gamma(2n-k+1),\\
  S_k(\Gamma(2n-k+1))&=\frac1{2n{-}k}\Gamma(2n-k+1).
\end{alignat*}
Throughout the paper the symbols $p$, $x$, $y$, $M$, $a_m$, $a_m'$, $a_m''$, \dots will be used
with the meaning they have in~\eqref{eq:hgdef}.
For fitting long formulas into the narrow columns of this layout, we also use the abbreviations
\begin{align*}
    A_m:=a_mn+a_m'k+a_m'', \quad&B_m:=b_mn+b_m'k+b_m'',\\
    U_m:=u_mn+u_m'k+u_m'', \quad&V_m:=v_mn+v_m'k+v_m''.
\end{align*}

We follow the paradigm of creative telescoping. For a given hypergeometric
term~$h$ as above, we want to determine polynomials
$\ell_0,\dots,\ell_r\in\K[n]$ (free of~$k$, not all zero), and a rational
function $C\in\K(n,k)$ (possibly involving~$k$, possibly zero), such that
\[
  \ell_0 h + \ell_1 S_n(h) + \cdots + \ell_r S_n^r(h) = S_k(C h) - C h.
\]
In this case, the operator $L:=\ell_0 + \ell_1 S_n + \cdots + \ell_r
S_n^r\in\K[n][S_n]$ is called a \emph{telescoper} for~$h$, and the rational
function $C\in\K(n,k)$ is called a \emph{certificate} for~$L$ (and~$h$). The
number~$r$ is called the \emph{order} of~$L$, and $d:=\max_{i=0}^r \deg_n\ell_i$
is called its \emph{degree}. If $h$ represents an actual sequence $f(n,k)$, then
a recurrence for the definite sum $\sum_{k=0}^n f(n,k)$ can be obtained from
such a pair $(L,C)$ as explained in the literature on symbolic
summation~\cite{petkovsek97}. We shall not embark on the technical subtleties of
this correspondence here but restrict ourselves to analyzing of the set of all pairs
$(r,d)$ for which there exists a telescoper of order $r$ and degree~$d$.

The following notation will be used.
\begin{itemize}
\item For $p\in\K[n,k]$ and $m\in\set N$, let
  \begin{alignat*}1
    \rf pm&:=p(p+1)(p+2)\cdots(p+m-1)
  \end{alignat*}
  with the conventions 
  $\rf p0=1$ and $\rf p1=p$.
\item For $p\in\K[n,k]$, $\deg_n p$ and $\deg_k p$ denote the degree of $p$ with respect
  to $n$ or~$k$, respectively. $\deg p$ without any subscript denotes the total degree of~$p$.
\item For $z\in\set R$, let $z^+:=\max\{0,z\}$.
\end{itemize}
With this notation, we have
\begin{alignat*}1
  \frac{S_n(h)}{h} &= x\frac{S_n(p)}{p}
    \prod_{m=1}^M \frac{\rf{A_m}{a_m}\rf{B_m}{b_m}}{\rf{U_m}{u_m}\rf{V_m}{v_m}},\\
  \frac{S_k(h)}{h} &= y\frac{S_k(p)}{p}
    \prod_{m=1}^M \frac{\rf{A_m}{a'_m}\rf{(V_m-v'_m)}{v'_m}}{\rf{U_m}{u'_m}\rf{(B_m-b'_m)}{b'_m}}.
\end{alignat*}

\section{The Non-Rational Case}\label{sec:trans}

We consider in this section the case where $h$ cannot be split into $h=q h_0$
for $q\in\K(n,k)$ and another hypergeometric term $h_0$ with
$S_k(h_0)/h_0=1$. Informally, this means that we exclude terms $h$ where $y=1$
and every $\Gamma$-term involving $k$ can be cancelled against another one to
some rational function. Those terms are treated separately in
Section~\ref{sec:rat} below.  If $h$ cannot be split as indicated, then also
$Ch$ cannot be split in this way, for any rational function $C\in\K(n,k)$. In
particular, we can then not have $Ch\in\K$ and therefore we always have
$S_k(Ch)-Ch\neq0$. This implies that whenever we have a pair $(L,C)$ with
$L(h)=S_k(Ch)-Ch$, we can be sure that $L$ is not the zero operator, and we need
not worry about this requirement any further.

The analysis in the present case is similar to that carried out by
Apagodu and Zeilberger~\cite{mohammed05}, who used it for deriving a bound on the order
$r$ of~$L$, and similar to our analysis~\cite{chen11} of the differential case.
The main idea is to follow step by step the execution of Zeilberger's
algorithm when applied to~$h$. This eventually leads to a linear system of
equations with coefficients in~$\K(n)$ which must have a solution whenever it is
underdetermined. The condition of having more variables than equations in this
linear system is the source of the estimate for choices $(r,d)$ that lead
to a solution.

\subsection{Zeilberger's Algorithm}

Recall the main steps of Zeilberger's algorithm: for some choice of~$r$, it makes
an ansatz $L=\ell_0+\ell_1 S_n+\cdots+\ell_r S_n^r$ with undetermined coefficients
$\ell_0,\dots,\ell_r$, and then calls Gosper's algorithm on~$L(h)$.
Gosper's algorithm~\cite{gosper78} proceeds by writing
\[
  \frac{S_k(L(h))}{L(h)}=\frac{S_k(P)}{P} \frac{Q}{S_k(R)}
\]
for some polynomials $P,Q,R$ such that $\gcd(Q,S_k^i(R))=1$ for all $i\in\set N$. It
turns out that the undetermined coefficients $\ell_0,\dots,\ell_r$ appear
linearly in~$P$ and not at all in $Q$ or~$R$. Next, the algorithm searches for
a polynomial solution $Y$ of the \emph{Gosper equation}
\[
  P = Q\,S_k(Y) - R\,Y
\]
by making an ansatz $Y=y_0+y_1k+y_2k^2+\cdots+y_sk^s$ for some suitably chosen
degree~$s$, substituting the ansatz into the equation, and comparing powers of $k$
on both sides. This leads to a linear system in the variables
$\ell_0,\dots,\ell_r,y_0,\dots,y_s$ with coefficients in~$\K(n)$. Any solution
of this system gives rise to a telescoper~$L$ with the corresponding certificate $C=RY/P$.
If no solution exists, the procedure is repeated with a greater value
of~$r$.

For a hypergeometric term~$h$ and an operator $L=\ell_0+\ell_1
S_n+\cdots+\ell_r S_n^r$, we have {\allowdisplaybreaks
\begin{alignat*}1
  L(h)&=\sum_{i=0}^r \ell_i x^i \frac{S_n^i(p)}p \prod_{m=1}^M \frac{\rf{A_m}{ia_m}\rf{B_m}{ib_m}}{\rf{U_m}{iu_m}\rf{V_m}{iv_m}} h\\
   &= \frac{\displaystyle\sum\limits_{i=0}^r \ell_i x^i S_n^i(p) \prod\limits_{m=1}^M {P_{i, m}}}
           {\displaystyle p\prod_{m=1}^M\rf{U_m}{ru_m}\rf{V_m}{rv_m}} h\\
   &=\biggl(\sum_{i=0}^r \ell_i x^i S_n^i(p) \prod\limits_{m=1}^M {P_{i, m}} \biggr)\\
   &\qquad{}\times x^n y^k \prod_{m=1}^M \frac{\Gamma(A_m)\Gamma(B_m)}{\Gamma(U_m+ru_m)\Gamma(V_m+rv_m)},\\
\intertext{where}
  P_{i,m}&=\rf{A_m}{ia_m}\rf{B_m}{ib_m}\rf{(U_m+iu_m)}{(r-i)u_m}\rf{(V_m+iv_m)}{(r-i)v_m}.
\end{alignat*}}%
We can write
\[
  \frac{S_k(L(h))}{L(h)}=\frac{S_k(P)}{P}\frac{Q}{S_k(R)},
\]
where
\begin{alignat*}1
   P &= \sum_{i=0}^r \ell_i x^i S_n^i(p) \prod\limits_{m=1}^M {P_{i, m}},\\
   Q &= y \prod_{m=1}^M \rf{A_m}{a'_m}\rf{(V_m+rv_m-v'_m)}{v'_m}, \\
   R &= \prod_{m=1}^M \rf{(U_m+ru_m-u'_m)}{u'_m}\rf{B_m}{b'_m}.
\end{alignat*}
Depending on the actual values of the coefficients appearing in~$h$, this decomposition may or may not
satisfy the requirement $\gcd(Q,S_k^i(R))=1$ for all $i\in\set N$. But even if
it does not, it only means that we may overlook some solutions, but every
solution we find still gives rise to a correct telescoper and
certificate. Since we are interested only in bounding the size of the
telescopers of~$h$, it is sufficient to study under which circumstances the
Gosper equation
\[
  P = Q\,S_k(Y) - R\,Y
\]
with the above choice of $P,Q,R$ has a solution.

\subsection{Counting Variables and Equations}

Apagodu and Zeilberger~\cite{mohammed05} proceed from here by analyzing the linear system over
$\K(n)$ resulting from the Gosper equation for a suitable choice of the degree
of~$Y$. They derive a bound on~$r$ but give no information on the
degree~$d$. General bounds for the degrees of solutions of linear systems with
polynomial coefficients could be applied, but they turn out to overshoot quite
much. In particular, it seems difficult to capture the phenomenon that increasing $r$
may allow for decreasing $d$ using such general bounds.

We proceed differently. Instead of a coefficient comparison with respect to powers
of $k$ leading to a linear system over~$\K(n)$, we consider a coefficient comparison
with respect to powers of $n$ and $k$ leading to a linear system over~$\K$. This
requires us to make a choice not only for the degree of $Y$ in~$k$ but also for the
degree of $Y$ in $n$ as well as for the degrees of the $\ell_i$ ($i=0,\dots,r$) in~$n$.
For expressing the number of variables and equations in this system, it is helpful
to adopt the following definition.

\begin{defi}\label{def:greek}
  For a proper hypergeometric term $h$ as in~\eqref{eq:hgdef}, let {\allowdisplaybreaks
  \begin{alignat*}3
       \delta &=\deg p,\\
    \vartheta &=\max\Bigl\{\sum_{m=1}^M (a_m+b_m),\sum_{m=1}^M (u_m+v_m)\Bigr\},\\[-2pt]
      \lambda &=\sum_{m=1}^M(u_m+v_m),\\
          \mu &=\sum_{m=1}^M (a_m+b_m-u_m-v_m),\\[-2pt]
          \nu &=\max\Bigl\{\sum_{m=1}^M (a'_m+v'_m),\sum_{m=1}^M (u'_m+b'_m)\Bigr\}.\\
  \end{alignat*}}
\end{defi}

Note that these parameters are integers which only depend on $h$ but not on $r$ or~$d$.
Except for $\mu$, they are all nonnegative. Note also that we have $\lambda+\mu\geq0$
and $\vartheta=\lambda+\mu^+\geq|\mu|$.

\begin{lemma}\label{lem:1}
  Let $d_i:=\deg_n\ell_i$ ($i=0,\dots,r$). Then
  \begin{alignat*}1
    &\deg P \leq \delta + \lambda r + \max_{i=0}^r (d_i + i \mu).
  \end{alignat*}
  Furthermore, $\deg_k P\leq \delta + \vartheta r$.
\end{lemma}
\begin{proof}
  It suffices to observe that
  \[
    \deg P_{i,m} \leq ia_m + ib_m + (r-i)u_m + (r-i)v_m
  \]
  for all $m=1,\dots,M$ and all $i=0,\dots,r$. For the degree with respect to~$k$,
  observe also that $\deg_k\ell_i=0$ for all~$i$.
\end{proof}

We have some freedom in choosing the~$d_i$. The choice influences the number of variables
in the ansatz
\[
  L=\sum_{i=0}^r \sum_{j=0}^{d_i} \ell_{i,j} n^j S_n^i
\]
as well as the number of equations. We prefer to have many variables and few equations.
For a fixed target degree~$d$, the maximum possible number of variables is $(d+1)(r+1)$
by choosing $d_0=d_1=\cdots=d_r=d$. But this choice also leads to many equations.
A better balance between number of variables and number of equations is obtained by lowering
some of the $d_i$ with indices close to zero (if $\mu$ is negative) or with indices
close to~$r$ (if $\mu$ is positive). Specifically, we choose
\[
  d_i := d - \left\{\begin{array}{ll}
      (\nu+i-r)^+|\mu|&\quad\text{if $\mu\geq 0$}\\
      (\nu-i)^+|\mu|&\quad\text{if $\mu< 0$}.
      \end{array}\right.
\]
See~\cite[Ex.~11, Ex.~15.5 and the remarks after Thm.~14]{chen11} for a detailed motivation of
the corresponding choice in the differential case. The support of the ansatz for $L$ looks
as in the following diagram, where every term $n^j S_n^i$ is represented by a bullet at
position $(i,j)$:
\begin{center}
  \begin{picture}(80,102)(0,-10)
    \multiput(0,0)(0,10){9}{\circle*{3}}
    \multiput(10,0)(0,10){9}{\circle*{3}}
    \multiput(20,0)(0,10){9}{\circle*{3}}
    \multiput(30,0)(0,10){9}{\circle*{3}}
    \multiput(40,0)(0,10){9}{\circle*{3}}
    \multiput(50,0)(0,10){9}{\circle*{3}}
    \multiput(60,0)(0,10){7}{\circle*{3}}
    \multiput(70,0)(0,10){5}{\circle*{3}}
    \multiput(80,0)(0,10){3}{\circle*{3}}
    \put(0,-5){\line(0,-1){3}}\put(0,-15){\hbox to0pt{\hss $\scriptstyle i=0$\hss}}
    \put(80,-5){\line(0,-1){3}}\put(80,-15){\hbox to0pt{\hss $\scriptstyle i=r$\hss}}
    \put(-5,0){\line(-1,0){3}}\put(-12,-2){\hbox to0pt{\hss $\scriptstyle j=0$}}
    \put(-5,80){\line(-1,0){3}}\put(-12,78){\hbox to0pt{\hss $\scriptstyle j=d$}}
    \put(83,73){$\Big\}\ \scriptstyle \mu$}
    \put(83,51){$\quad \vdots$}
    \put(83,33){$\Big\}\ \scriptstyle \mu$}
    \put(56,85){$\scriptstyle\overbrace{\rule{25pt}{0pt}}^{\nu}$}
  \end{picture}
  \medskip
\end{center}
With this choice for the degrees~$d_i$, the number of variables in the ansatz for $L$ is
\[
 \sum_{i=0}^r (d_i+1) = (d+1)(r+1) - \tfrac12|\mu|\nu(\nu+1),
\]
provided that $d\geq|\mu|\nu$.
The number of resulting equations is as follows.

\begin{lemma}\label{lem:2}
  If the $d_i$ are chosen as above, then $P$ contains at most
  \[
    \tfrac12\bigl(\delta+\vartheta r+1\bigr)\bigl(\delta + 2d + \vartheta r - 2|\mu|\nu + 2\bigr)
  \]
  terms $n^i k^j$.
\end{lemma}
\begin{proof}
  If $\mu\geq 0$, we have
  \[
    d_i + i\mu = d - (\nu+i-r)^+\mu + i\mu \leq d - \nu|\mu| + r\mu
  \]
  for all $i=0,\dots,r$. Likewise, when $\mu<0$, we have
  \[
    d_i + i\mu = d - (\nu-i)^+|\mu| + i\mu \leq d - \nu|\mu|
  \]
  for all $i=0,\dots,r$. Together with Lemma~\ref{lem:1}, it follows that
  \[
    \deg P\leq \delta + (\lambda+\mu^+)r + d - \nu|\mu| = \delta + \vartheta r + d - \nu|\mu|
  \]
  regardless of the sign of~$\mu$. We also have $\deg_k P\leq\delta +\vartheta r$
  from Lemma~\ref{lem:1}. For the number of terms $n^i k^j$ in $P$ we have
  \[
   \sum_{i=0}^{\deg_k P}\! (1+\deg P - i) = \tfrac12(\deg_kP +1)(2\deg P + 2 - \deg_kP).
  \]
  Plugging the estimates for $\deg P$ and $\deg_k P$ into the right hand side gives
  the expression claimed in the Lemma.
\end{proof}

The support of $P$ has a trapezoidal shape which is determined by the total degree and 
the degree with respect to~$k$:

\begin{center}
  \begin{picture}(30,80)(0,-10)
    \multiput(0,0)(0,10){8}{\circle*{3}}
    \multiput(10,0)(0,10){7}{\circle*{3}}
    \multiput(20,0)(0,10){6}{\circle*{3}}
    \multiput(30,0)(0,10){5}{\circle*{3}}

    \put(-4,-5){\hbox{$\scriptstyle\underbrace{\rule{38pt}{0pt}}_{\deg_k P}$}}
    \put(-5,35){\hbox to 0pt{\hss$\scriptstyle\deg P\ \left\{\rule{0pt}{41pt}\right.$}}
  \end{picture}
  \medskip
\end{center}

The next step is to choose the degrees for $Y$ in $n$ and~$k$. This is done in such a way that
$Q\,S_k(Y)-R\,Y$ only contains terms which are already expected to occur in~$P$, so that no
additional equations will appear.

\begin{lemma}\label{lem:3}
  Let the $d_i$ be chosen as before and suppose that $Y\in\K[n,k]$ is such that
  $\deg Y\leq \deg P - \nu$ and $\deg_k Y\leq \deg_k P - \nu$. Then
  $P-(Q\,S_k(Y)-R\,Y)$ contains at most
  \[
    \tfrac12\bigl(\delta+\vartheta r+1\bigr)\bigl(\delta + 2d + \vartheta r - 2|\mu|\nu + 2\bigr)
  \]
  terms $n^i k^j$.
\end{lemma}
\begin{proof}
  As for Lemma~\ref{lem:2}, using also $\max\{\deg Q,\deg R\}=\max\{\deg_k Q,\deg_k R\}=\nu$.
\end{proof}

Lemma~\ref{lem:3} suggests the ansatz
\[
  Y = \sum_{i=0}^{s_1}\sum_{j=0}^{s_2-i} y_{i,j}k^i n^j
\]
with $s_1=\deg_k P-\nu$ and $s_2=\deg P-\nu$, which provides us with
\begin{alignat*}1
  &\tfrac12\bigl(\delta+\vartheta r+1-\nu\bigr)\bigl(\delta + 2d + \vartheta r - 2|\mu|\nu + 2-\nu\bigr)
\end{alignat*}
variables. We are now ready to formulate the main result of this section.
Note that the inequality for~$d$ is a considerably simpler formula
than the corresponding result in the differential case (Thm.~14 in~\cite{chen11}).

\begin{theorem}\label{thm:trans}
  Let $h$ be a proper hypergeometric term which cannot be written $h=qh_0$ for
  some $q\in\K(n,k)$ and a hypergeometric term $h_0$ with $S_k(h_0)/h_0=1$.
  Let $\delta,\lambda,\mu,\nu$ be as in Definition~\ref{def:greek}, let $r\geq\nu$ and
  \[
    d>\frac{\bigl(\vartheta\nu-1\bigr)r + \tfrac12\nu\bigl(2\delta+|\mu|+3-(1+|\mu|)\nu\bigr) - 1}{r-\nu+1}.
  \]
  Then there exists a telescoper $L$ for $h$ of order~$r$ and degree~$d$.
\end{theorem}
\begin{proof}
  A sufficient condition for the existence of a telescoper of order $r$ and degree~$d$ is that
  for some particular ansatz, the equation
  \[
    P=Q\,S_k(Y)-R\,Y
  \]
  has a nontrivial solution. A sufficient condition for the existence of a solution is that the
  linear system resulting from coefficient comparison has more variables than equations.
  For all $d$ in question, we have $d>\vartheta\nu\geq|\mu|\nu$. Therefore,
  with the ansatz described above, we have
  \[
    (d+1)(r+1)-\tfrac12|\mu|\nu(\nu+1)
  \]
  variables $\ell_{i,j}$ in~$P$,
  \begin{alignat*}1
    \tfrac12\bigl(\delta+\vartheta r+1-\nu\bigr)\bigl(\delta + 2d + \vartheta r - 2|\mu|\nu + 2-\nu\bigr)
  \end{alignat*}
  variables $y_{i,j}$ in~$Y$, and
  \[
    \tfrac12\bigl(\delta+\vartheta r+1\bigr)\bigl(\delta + 2d + \vartheta r - 2|\mu|\nu + 2\bigr)
  \]
  equations. Solving the inequality
  \begin{alignat*}1
    &(d+1)(r+1)-\tfrac12|\mu|\nu(\nu+1)\\
    &\quad{}+\tfrac12\bigl(\delta+\vartheta r+1-\nu\bigr)\bigl(\delta + 2d + \vartheta r - 2|\mu|\nu + 2-\nu\bigr)\\
    &>\tfrac12\bigl(\delta+\vartheta r+1\bigr)\bigl(\delta + 2d + \vartheta r - 2|\mu|\nu + 2\bigr)
  \end{alignat*}
  under the assumption $r\geq\nu$ for $d$ gives the claimed degree estimate.
\end{proof}

\subsection{Examples and Consequences}

\begin{example}\label{ex:trans}
\begin{enumerate}
\item For
  \[
    h=(n^2+k^2+1)\frac{\Gamma(2n+3k)}{\Gamma(2n-k)}
  \]
  we have $\delta=2$, $\vartheta=2$, $\mu=0$, $\nu=4$. Theorem~\ref{thm:trans} predicts
  a telescoper of order~$r$ and degree~$d$ whenever $r\geq4$ and
  \[
    d>\frac{7r+5}{r-3}.
  \]
  The left figure below shows the curve defined by the right hand side (black) together with
  the region of all points $(r,d)$ for which we found telescopers of $h$ with order~$r$
  and degree~$d$ by direct calculation (gray).
  In this example, the estimate overshoots by very little only.
\item The corresponding picture for
  \[
    h=\frac{\Gamma(2n+k)\Gamma(n-k+2)}{\Gamma(2n-k)\Gamma(n+2k)}
  \]
  is shown below on the right. Here, $\delta=0$, $\vartheta=3$, $\mu=0$, $\nu=3$ and
  Theorem~\ref{thm:trans} predicts a telescoper of order~$r$ and degree~$d$ whenever $r\geq3$
  and
  \[
    d>\frac{8r-1}{r-2}.
  \]
  In this example, the estimate is less tight.



\end{enumerate}
\end{example}

\centerline{\includegraphics{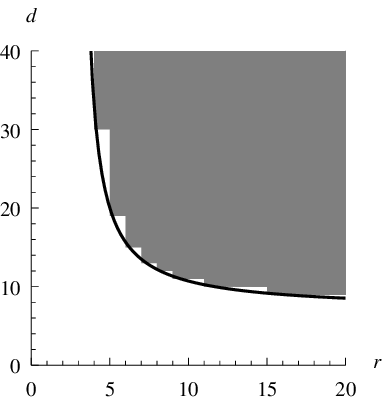}\hfil\includegraphics{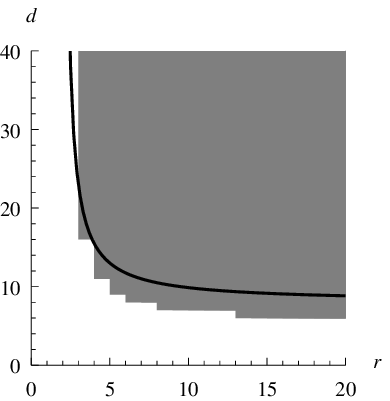}}

\medskip

The points $(r,d)$ in the portion of the gray region which is below the black curve
represent telescopers where the corresponding linear system resulting from the ansatz considered
in our proof is overdetermined but, for some strange reason, nevertheless nontrivially solvable. 
The small portions of white space which lie above the curves are not in contradiction with our 
theorem because they do not contain any points with integer coordinates. (The theorem says that
every point $(r,d)\in\set Z^2$ above the curve belongs to the gray region.)

Theorem~\ref{thm:trans} supplements the bound given in \cite{mohammed05} on the order
of telescopers for a hypergeometric term by an estimate for the degree that these
operators may have. In addition, it provides lower degree bounds for higher orders
and admits a bound on the least possible degree for a telescoper.

\begin{corollary}
  With the notation of Theorem~\ref{thm:trans}, $h$ admits a telescoper of
  order $r=\nu$ and degree
  \[
    d=\bigl\lceil\tfrac12\nu(2\delta+2\nu\vartheta+|\mu|-\nu|\mu|)\bigr\rceil
  \]
  as well as a telescoper of order
  \[
    r = \bigl\lceil\tfrac12\nu(1+2\delta+2(\nu-1)(\vartheta-|\mu|))\bigr\rceil
  \]
  and degree $d=\vartheta\nu$.
\end{corollary}
\begin{proof}
  Immediate by checking that the two choices for $r$ and $d$ satisfy the conditions
  stated in Theorem~\ref{thm:trans}.
\end{proof}

An accurate prediction for the degrees of the telescopers can also be used for improving the efficiency
of creative telescoping algorithms. Although most implementations today compute the telescoper
with minimum order, it may be less costly to compute a telescoper of slightly higher order. If
we know in advance the degrees $d$ of the telescopers for every order~$r$, we can select before
the computation the order $r$ which minimizes the computational cost. Of course, the cost depends
on the algorithm which is used. It is not necessary (and not advisable) to follow the steps in the
derivation of Theorem~\ref{thm:trans} and do a coefficient comparison over~$\K$. Instead, one should
follow the common practice~\cite{koutschan10b} of comparing coefficients only with respect to powers of $k$
and solve a linear system over~$\K(n)$. For nonminimal choices of $r$, this system will have a
nullspace of dimension greater than one, of which we need not compute a complete basis,
but only a single vector with components of low degree. There are algorithms known for computing
such a vector using $\O(m^3t)$ field operations when the system has at most $m$ variables and equations and the
solution has degree at most $t$ \cite{beckermann94,storjohann05,bostan07}. In the situation at hand, we have
$m=(r+1)+(\delta+\vartheta r+1)$ variables and a solution of degree $t=\delta+\vartheta r+d-(|\mu|+1)\nu+1$.
Therefore, in order to compute a telescoper
and its certificate most efficiently, we should minimize the cost function
\[
 C(r,d):=\bigl((\vartheta+1)r+\delta+2\bigr)^3\bigl(\delta+\vartheta r+d-(|\mu|+1)\nu+1\bigr).
\]
According to the following theorem, for asymptotically large input it is significantly better to
choose $r$ slightly larger than the minimal possible value.

\begin{theorem}\label{thm:trans:compl}
  Let $h$ and $\lambda,\mu,\nu$ be as in Theorem~\ref{thm:trans}, $\tau\geq\max\{\vartheta, \nu\}$,
  and suppose that $\kappa\in\set R$ is a constant such that degree $t$ solutions of a linear system
  with $m$ variables and at most $m$ equations over $\K(n)$ can be computed with $\kappa m^3 t$ operations
  in~$\K$. Then:
  \begin{enumerate}
  \item A telescoper of order $r=\tau$ along with a corresponding certificate can be computed using
  \[
    \kappa\tau^9 + \tfrac12(7-|\mu|)\kappa\tau^8 + \O(\tau^7)
  \]
  operations in $\K$.
  \item\label{it:2} If $\alpha>1$ is some constant and $r$ is chosen such that
  $r=\alpha \tau + \O(1)$,
  then a telescoper of order $r$ and a corresponding certificate can be computed using
  \[
    \frac{\alpha^5}{\alpha -1}\kappa\tau^8 + \O(\tau^7)
  \]
  operations in~$\K$.
  \end{enumerate}
  In particular, a telescoper for $h$ and a corresponding certificate can be computed
  in polynomial time.
\end{theorem}
\begin{proof}
  According to Theorem~\ref{thm:trans}, for every $r\geq\tau$ there exists a telescoper of
  order $r$ and degree $d$ for any
  \[
    d > f(r):=\frac{(\tau^2-1)r + \O(\tau^2)}{r-\tau-1}.
  \]
  By assumption, such a telescoper can be computed using no more than
  \[
    C(r,d)=\kappa \bigl((\tau+1)r+\delta+2\bigr)^3\bigl(\delta+\tau r+d-(|\mu|+1)\tau+1\bigr)
  \]
  operations in~$\K$. The claim now follows from the asymptotic expansions of
  $C(\tau,f(\tau)+1)$ and $C(\alpha\tau,f(\alpha\tau)+1)$ for $\tau\to\infty$, respectively.
\end{proof}

The leading coefficient in part~\ref{it:2} is minimized for $\alpha=5/4$. This suggests that
when $\vartheta$ and $\nu$ are large and approximately equal, the order of the cheapest telescoper
is about 20\% larger than the minimal expected order.

\section{The Rational Case}\label{sec:rat}

We now turn to the case where $h$ can be written as $h=qh_0$ for some hypergeometric term $h_0$
with $S_k(h_0)/h_0=1$. By the following transformation, we may assume without loss of generality $h_0=1$.

\begin{lemma}
  Let $h$ be a hypergeometric term and suppose that $h=qh_0$ for some $q\in\K(n,k)$ and a hypergeometric
  term $h_0$ with $S_k(h_0)/h_0=1$. Let $a,b\in\K[n,k]$ be such that $S_n(h_0)/h_0=a/b$.
  Let $L$ be a telescoper for $q$ of order $r$ and degree~$d$.
  Then there exists a telescoper for $h$ of order $r$ and degree at most
  $d+r\max\{\deg_n a, \deg_n b\}$.
\end{lemma}
\begin{proof}
  Write $L=\ell_0+\ell_1S_n+\cdots+\ell_rS_n^r$ and let $C\in\K(n,k)$ be a certificate
  for~$L$ and~$q$, so $L(q)=S_k(Cq)-Cq$. For $i=0,\dots,r$, let
  \[
    \tilde\ell_i := \ell_i \frac ba S_n\Bigl(\frac ba\Bigl)\cdots S_n^{i-1}\Bigl(\frac ba\Bigr)
  \]
  and $\tilde L:=\tilde\ell_0+\tilde\ell_1 S_n + \cdots + \tilde\ell_rS_n^r$. Then
  \[
    \tilde L(qh_0) = L(q)h_0 = (S_k(Cq)-Cq)h_0 = S_k(Cqh_0)-Cqh_0.
  \]
  Because of
  \[
    S_k\Bigl(\frac ab\Bigr) = \frac{S_k(S_n(h_0))}{S_k(h_0)}=
    \frac{S_n(S_k(h_0))}{S_k(h_0)}=\frac{S_n(h_0)}{h_0}=\frac ab,
  \]
  the operator $\tilde L$ is free of~$k$. Thus, after clearing denominators, $\tilde L$ is
  a telescoper for $h$ with coefficients of degree at most $d+r\max\{\deg_n a,\deg_n b\}$.
\end{proof}

{}From now on, we assume that $h$ is at the same time a proper hypergeometric term and a rational
function, or equivalently, that $h$ is a rational function whose denominator factors into integer-linear
factors. Le~\cite{le03} gives a precise description of the structure of telescopers in this case, and he
proposes an algorithm different from Zeilberger's for computing them. Our degree estimate is derived
following the steps of his algorithm, so we start by briefly summarizing the main steps of Le's approach.

\subsection{Le's Algorithm}\label{sec:le}

Given a rational proper hypergeometric term~$h$, Le's algorithm computes a telescoper $L$ for $h$ as follows.

\begin{enumerate}
\item Compute $g\in\K(n,k)$ and polynomials $p,q\in\K[n,k]$ with $\gcd(q,S_k^i(q))=1$ for all $i\in\set Z\setminus\{0\}$
  such that
  \[
    h = S_k(g) - g + \frac pq.
  \]
  Then an operator $L$ is a telescoper for $h$ if and only if $L$ is a telescoper for $\frac pq$.
  Abramov~\cite{abramov95a} and Paule~\cite{paule95} explain how to compute such a decomposition.
\item Compute a polynomial $u\in\K[n]$, operators $V_1,\dots,V_s$ in $\K[n][S_n]$, and
  rational functions $f_1,\dots,f_s$ of the form $f_i=(a_in+a_i'k+a_i'')^{-e_i}$ ($i=1,\dots,s$) such that
  \[
    \frac pq = \frac1u\sum_{i=0}^s V_i(f_i).
  \]
  Such data always exists according to Lemma~5 in \cite{le03} in combination with the assumption $\gcd(q,S_k^i(q))=1$
  ($i\in\set Z\setminus\{0\}$). It can be further assumed that the $f_i$ are chosen such that $a_i'>0$, $e_i>0$,
  $\gcd(a_i,a_i')=1$ for all~$i$, and
  \[
    \Bigl(\frac{a_i}{a_i'} - \frac{a_j}{a_j'}\Bigr)n +
    \Bigl(\frac{a_i''}{a_i'} - \frac{a_j''}{a_j'}\Bigr) \not\in\set Z
  \]
  for all $i\neq j$ with $e_i=e_j$.
\item For $i=1,\dots,s$, compute an operator $L_i\in\K(n)[S_n]$ such that $S_n^{a_i'}-1$ is a right divisor
  of $L_i (\frac1uV_i)$. It follows from Le's Lemma~4 that the operators $L_i$ with this property are precisely
  the telescopers of the rational functions~$V_i(f_i)$.
\item Compute a common left multiple $L\in\K[n][S_n]$ of the operators $L_1,\dots,L_s$.
  Then $L$ is a telescoper for~$h$.
\end{enumerate}

The main part of the computational work happens in the last two steps. It
therefore appears sensible to assume in the following degree analysis that we
already know the data $u, V_1,\dots,V_s, f_1,\dots,f_s$ computed in step~2, and
to express the degree bounds in terms of their degrees and coefficients rather
than in terms of the degrees of numerator and denominator of~$h$, say.

\subsection{Counting Variables and Equations}

Also in the present case, the degree estimate is obtained by balancing the number of variables and
equations of a certain linear system over~$\K$. The linear system we consider originates from a
particular way of executing steps 3 and~4 of the algorithm outlined above.

\begin{theorem}\label{thm:curve:rat}
  Let $u\in\K[n]$ and let $V_1,\dots,V_s\in\K[n][S_n]$
  be operators of degree~$\delta_i$ ($i=1\dots,s$).
  Let $f_i=(a_in+a_i'k+a_i'')^{-e_i}$ for some $a_i''\in\K$,
  $a_i,a_i'\in\set Z$ with $a_i'>0$ and $\gcd(a_i,a_i')=1$, $e_i>0$, suppose
  \[
    \Bigl(\frac{a_i}{a_i'} - \frac{a_j}{a_j'}\Bigr)n +
    \Bigl(\frac{a_i''}{a_i'} - \frac{a_j''}{a_j'}\Bigr) \not\in\set Z
  \]
  for all $i\neq j$ with $e_i=e_j$. Let $h=\frac1u\sum_{i=1}^s V_i(f_i)$.
  Then for every $r\geq\sum_{i=1}^s a_i'$ and every
  \[
    d>\frac{\displaystyle-r-1+\sum_{i=1}^s a_i'\delta_i}
           {\displaystyle r+1-\sum_{i=1}^sa_i'} + \deg_n u
  \]
  there exists a telescoper $L$ for $h$ of order~$r$ and degree~$d$.
\end{theorem}
\begin{proof}
  According to Le's algorithm, it suffices to find some $L\in\K[n][S_n]$ and operators
  $R_i\in\K(n)[S_n]$ with the property that $L(\frac1uV_i)= R_i(S_n^{a_i'}-1)$ for all~$i$.

  Denote by $\rho_i$ the order of~$V_i$.  Writing $\tilde d:=d-\deg_n u$, we
  make an ansatz $L=\tilde L u$ with
  \[
   \tilde L=\sum_{i=0}^r \sum_{j=0}^{\tilde d} \ell_{i,j}n^j S_n^i
  \]
  so that $L$ has degree~$d$ and $L \frac1uV_i = \tilde L V_i$ ($i=1,\dots,s$).
  It thus remains to construct operators $R_i\in\K[n][S_n]$ with $\tilde L V_i = R_i(S_n^{a_i'}-1)$.
  Since $L\tilde V_i$ has order $r+\rho_i$ and degree $\tilde d+\delta_i$, we consider
  ansatzes for the $R_i$ of order $r+\rho_i-a_i'$ and degree $\tilde d+\delta_i$, respectively,
  because $S_n^{a_i'}-1$ has order $a_i'$ and degree~$0$.
  Then we have altogether
  \[
    (r+1)(\tilde d+1) + \sum_{i=1}^s (r+\rho_i-a_i'+1)(\tilde d+\delta_i+1)
  \]
  variables in $\tilde L$ and the~$R_i$, and comparing coefficients with respect to $n$ and $S_n$
  in all the required identities $\tilde L V_i=R_i(S_n^{a_i'}-1)$ leads to a linear system with
  \[
    \sum_{i=1}^s (r+\rho_i+1)(\tilde d+\delta_i+1)
  \]
  equations. This system will have a nontrivial solution whenever the number of variables
  exceeds the number of equations. Under the assumption $r\geq\sum_{i=1}^s a_i'$, the inequality
  \begin{alignat*}1
    & (r+1)(\tilde d+1)+ \sum_{i=1}^s (r+\rho_i-a_i'+1)(\tilde d+\delta_i+1) \\
    &\qquad{} > \sum_{i=1}^s (r+\rho_i+1)(\tilde d+\delta_i+1)
  \end{alignat*}
  is equivalent to
  \[
    \tilde d > \frac{-r-1 + \sum_{i=1}^s a_i'\delta_i}{r+1-\sum_{i=1}^sa_i'}.
  \]
  This completes the proof.
\end{proof}

\subsection{Examples and Consequences}

\begin{example}\label{ex:rat}
  \begin{enumerate}
  \item The rational function
    \[
      h = \frac{(2n-3k)(3n-2k)^2}{(n+k+2)(n+2k+1)(2n+k+1)(3n+k+1)}
    \]
    can be written in the form $h=\frac1 u\sum_{i=1}^4 V_i(f_i)$ where
    $u=(n-1)n(n+3)(2n-1)(3n+1)(5n+1)$,
    the $f_i$ are such that $a_1'=a_2'=a_3'=1$, $a_4'=2$, and
    the $V_i$ are such that $\delta_1=\cdots=\delta_4=6$.
    Therefore, Theorem~\ref{thm:curve:rat} predicts a telescoper of order~$r$ and degree~$d$
    whenever $r\geq5$ and
    \[
      d > \frac{29-r}{r-4} + 6.
    \]
    This curve together with the region of all points $(r,d)$ for which 
    a telescoper of order~$r$ and degree~$d$ exists is shown in the left figure below.
  \item The corresponding picture for the rational function
    \[
      h = \frac{(n-k+1)^2(2n-3k+5)}{(n+k+3)(n+k+5)(n+2k+1)(2n+k+1)^2}
    \]
    is shown in the figure below on the right. This input can be written in the form
    \[
      h = S_k(g) - g + \frac1u\sum_{i=1}^4 V_i(f_i)
    \]
    with $g=\frac{10(n+3)^2(n+4)(2n+2k+7)}{(n-4)^2(n+9)(n+k+3)(n+k+4)}$,
    $u=(3n+1)^2(n-4)^2(n-2)^2(n+5)(n+9)$,
    the $f_i$ such that $a_1'=a_2'=a_3'=1$, $a_4'=2$, and
    the $V_i$ such that $\delta_1=8$, $\delta_2=\delta_3=\delta_4=7$.
    According to Theorem~\ref{thm:curve:rat}, we therefore expect a telescoper for $h$
    of order~$r$ and degree~$d$ whenever $r\geq5$ and
    \[
      d > \frac{35-r}{r-4} + 8.
    \]
    In this example, the estimate is not as tight as in the previous one.
  \end{enumerate}
\end{example}



\centerline{\includegraphics{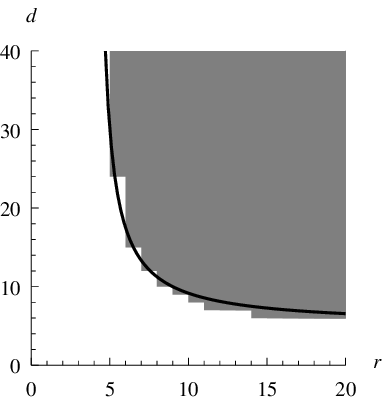}\hfil\includegraphics{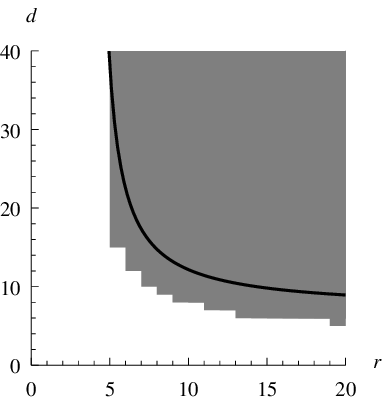}}

\smallskip

Again, it is an easy matter to specialize the general degree bound to a degree estimate
for a low order telescoper, or to an order estimate for a low degree telescoper.

\begin{corollary}
  With the notation of Theorem~\ref{thm:curve:rat}, $h$ admits a telescoper of
  order $r=\sum_{i=1}^s a_i'$ and degree $d=\deg_n u + \sum_{i=1}^s (\delta_i-1)a_i'$
  as well as a telescoper of order $r=\sum_{i=1}^s a_i\delta_i$ and degree $d=\deg_nu$.
\end{corollary}
\begin{proof}
  Clear by checking that the proposed choices for $r$ and~$d$ are consistent with
  the bounds in Theorem~\ref{thm:curve:rat}.
\end{proof}

Also like in the non-rational case, the bounds for the degrees of the telescopers can be used for
deriving bounds on the computational cost for computing them. In the present situation,
let us assume for simplicity that the cost of steps 1 and~2 of Le's algorithm is negligible,
or equivalently, that the input~$h$ is of the form $\frac1u\sum_{i=0}^s V_i(f_i)$ with $V_i\in\K[n][S_n]$.
We shall analyze the algorithm which carries out steps 3 and~4 of Section~\ref{sec:le} in one
stroke by making an ansatz over $\K(n)$ for an operator~$L=\ell_0+\ell_1S_n+\cdots+\ell_rS_n^r$, computing
the right reminders of $L\frac1uV_i$ with respect to $S_n^{a_i'}-1$ and equating their coefficients to
zero. We assume, as before, that the resulting linear system is solved using an
algorithm whose runtime is linear in the output degree and cubic in the matrix size. Then
the algorithm requires $\O(r^3 d)$ operations in~$\K$.

\begin{theorem}\label{thm:rat:compl}
  Let $u\in\K[n]$, $V_1,\dots,V_s\in\K[n][S_n]$, and $f_1,\dots,f_s\in\K(n,k)$ be as in Theorem~\ref{thm:curve:rat}
  and consider $h=\frac1u\sum_{i=1}^s V_i(f_i)$.
  Suppose that $\kappa\in\set R$ is a constant such that degree $t$ solutions of a linear system
  with $m$ variables and at most $m$ equations over $\K(n)$ can be computed with $\kappa m^3 t$ operations
  in~$\K$. Assume $\delta_1=\cdots=\delta_s=:\delta>0$ and $a_1'=a_2'=\cdots=a_s'=:a'$ are fixed.
  Then:
  \begin{enumerate}
  \item A telescoper of order $r=a's$ can be computed using
  \[
    a'{}^4 (\delta-1) \kappa\,s^4 + \O(s^3)
  \]
  operations in $\K$.
  \item\label{it:rat:2} If $\alpha>1$ is some constant and $r$ is chosen such that
  $r=\alpha a's + \O(1)$
  then a telescoper of order~$r$ can be computed using
  \[
   \frac{\alpha^3}{\alpha-1}a'{}^3 (\delta-1+(\alpha-1)\deg_nu)\kappa\, s^3 + \O(s^2)
  \]
  operations in~$\K$.
  \end{enumerate}
  In particular, a telescoper for $h$ can be computed in polynomial time.
\end{theorem}
\begin{proof}
  According to the first estimate stated in Theorem~\ref{thm:curve:rat}, for every $r\geq a's$
  there exists a telescoper of order~$r$ and degree~$d$ for any
  \[
    d > f(r) := \frac{sa'\delta-r-1}{r+1-sa'} + \deg_nu.
  \]
  By assumption, such a telescoper can be computed using no more than $C(r,d) := \kappa r^3 d$
  operations in~$\K$. The claim now follows from the asymptotic expansions of $C(a's, f(a's)+1)$
  and $C(\alpha a's, f(\alpha a's)+1)$ for $s\to\infty$, respectively.
\end{proof}

When $\deg_nu=0$, the leading coefficient in part~\ref{it:rat:2} is minimized
for $\alpha=3/2$. This suggests that when $s$ is large and all the $\delta_i$,
and $a_i'$ are approximately equal, the order of the cheapest operator
exceeds the minimal expected order by around~50\%.

It must not be concluded from a literal comparison of the exponents in
Theorems~\ref{thm:trans} and~\ref{thm:curve:rat} that Le's algorithm is faster
than Zeilberger's, because $\tau$ in Theorem~\ref{thm:trans} and $s$ in
Theorem~\ref{thm:curve:rat} measure the size of the input differently. Nevertheless,
it is plausible to expect that Le's algorithm is faster, because it finds
the telescopers without also computing a (potentially big) corresponding
certificate. Our main point here is not a comparison of the two approaches, but
rather the observation that both of them admit a degree analysis which fits to the
general paradigm that increasing the order can cause a degree drop which is
significant enough to leave a trace in the computational complexity.

It can also be argued that the situations considered in
Theorems~\ref{thm:trans:compl} and~\ref{thm:rat:compl} are chosen somewhat
arbitrarily ($\vartheta$~and $\nu$ growing while $\mu$~remains fixed; resp.\
$s$~growing while all the $\delta_i$ and $a_i'$ remain fixed). Indeed, it would
be wrong to take these theorems as an advice which telescopers are most easily
computed for a particular input at hand. Instead, in order to speed up an actual
implementation, one should let the program calculate the optimal choice for~$r$
from the degree estimates given Theorems~\ref{thm:trans} and~\ref{thm:curve:rat}
with the particular parameters of the input.

Unfortunately, we are not able to illustrate the speedup obtained in this way by
an actual runtime comparison for a concrete example, because for examples which
can be handled on currently available hardware, the computational cost turns out
to be minimized for the least order operator. But already for examples which are
only slightly beyond the capacity of current machines, the degree predictions in
Theorems~\ref{thm:trans} and~\ref{thm:curve:rat} indicate that computing the
telescoper of order one more than minimal will start to give an advantage. We
therefore expect that the results presented in this paper will contribute to the
improvement of creative telescoping implementations in the very near future.


\end{document}